\documentclass[aps,12pt,reprint,groupedaddress,superscriptaddress,onecolumn,longbibliography]{revtex4-2}
\usepackage[usenames,dvipsnames]{xcolor}
\usepackage{amssymb}
\usepackage{graphicx}
\usepackage{amsmath}
\usepackage{amsthm}
\usepackage{dsfont}
\usepackage[bookmarks=true,colorlinks,citecolor=blue,urlcolor=blue]{hyperref}
\usepackage{braket}
\usepackage{babel}
\usepackage{blindtext}
\usepackage{physics}
\usepackage{bm}
\newtheorem{theorem}{Theorem}[section]

\newtheorem{proposition}[theorem]{Proposition}

\begin{document}

\title{Quantum-informed surrogate sampling for combinatorial optimization}

\author{Elisabeth Wybo}
\email{elisabeth.wybo@iqm.tech}
\author{Jernej Rudi Fin\v{z}gar}
\affiliation{IQM Quantum Computers, Georg-Brauchle-Ring 23-25, 80992 München, Germany}

\begin{abstract}
    We introduce Quantum-Informed Surrogate Sampling (QISS), a post-processing framework that generates candidate solutions to combinatorial optimization problems from low-weight correlations of shallow quantum circuits.
    The quantum device estimates local observables, which are directly accessible by repeated measurements and for which a wide range of error-mitigation tools are available, while candidate solutions are generated classically without explicit dependence on the combinatorial optimization problem itself.
    We evaluate QISS on Maximum Cut and Maximum Independent Set problems on $N$ variables and show that only $O(N)$ low-order correlators from shallow circuits suffice to produce competitive solutions that surpass vanilla QAOA.
    For MaxCut on 3-regular graphs, QISS from $p=3$ QAOA correlators outperforms vanilla QAOA at $p=17$ on average, with further improvements possible by warm-starting QAOA.
    We validate the procedure on the 54-qubit IQM Emerald quantum device and demonstrate its noise resilience.
    Our results support a regime for near-term optimization in which shallow circuits serve not as direct samplers but as generators of informative statistics for scalable classical sampling.
\end{abstract}

\maketitle

\section{Introduction}
Combinatorial optimization lies at the heart of many problems in science and engineering, yet many relevant problems remain computationally extremely challenging at scale. Quantum computers offer a qualitatively different computational paradigm, motivating the development of quantum algorithms and hybrid quantum--classical heuristics that may ultimately complement---and, in favorable regimes, potentially surpass---classical optimization methods~\cite{Hogg2000,Hogg2000a,Farhi2014,Montanaro2020,Abbas2023,Dalzell2023,Jordan2025,Herman2025}.
Among these approaches, the Quantum Approximate Optimization Algorithm (QAOA)~\cite{Farhi2014} has emerged as a leading candidate because of its simple alternating-operator structure and its compatibility with relatively shallow circuit implementations~\cite{Pagano2020,Harrigan2021,Ebadi2022,Pelofske2024,Shaydulin2024,He2025}.

A growing body of evidence suggests that QAOA can compete with or even outperform classical solvers in certain regimes~\cite{Boulebnane2024,Boulebnane2024a,Shaydulin2024,Montanaro2025}. However, on near-term hardware, noise and coherence limitations typically restrict practical QAOA implementations to shallow circuits with constant depth. Such circuits are local for sparse instances of bounded degree: any local observable is determined by a bounded reverse light cone whose radius is set by the depth $p$, not by the problem size~\cite{Akshay2021,Farhi2020a,Chou2021,Basso2022,Anshu2022}. 
For optimization, this locality can be a genuine limitation~\cite{Farhi2020a,Farhi2020}: fixed-$p$ QAOA may fail to exploit global problem structure and can inherit known limitations of local classical algorithms on families of sparse graphs~\cite{Chou2021,Basso2022,Chen2023}. Related average-case obstructions can arise in random optimization landscapes, where the geometry of near-optimal solutions (e.g. overlap-gap phenomena) is known to restrict broad classes of efficient algorithms~\cite{Gamarnik2013,Gamarnik2019,Gamarnik2021}.

 However, bounded light cones also allow shallow circuits to estimate low-order statistics efficiently and with relatively low sensitivity to noise~\cite{Franca2020,DePalma2022}. We therefore propose a different division of labor: rather than asking QAOA to produce good solutions directly, we treat it as a source of low-order statistics and delegate the generation of candidate solutions to classical post-processing. Low-weight expectation values are a natural target, since they are amenable to error mitigation, can be estimated to fixed precision with a sample cost independent of system size, and, because each depends only on a bounded subgraph, can be evaluated locally even when the full optimization instance does not fit on the quantum processor~\cite{Dupont2025,Wybo2026}. The classical cost of reproducing these statistics is controlled by the treewidth of the circuit's light cones, which is exponential in the worst case but bounded at the shallow depths and bounded connectivity we consider~\cite{Shi2008}. In contrast, the quantum estimation cost is set by the operator norm and circuit depth and is insensitive to treewidth, whereas classical contraction becomes infeasible as depth or connectivity grow.

\begin{figure*}[t]
    \centering
    \includegraphics[width=\textwidth]{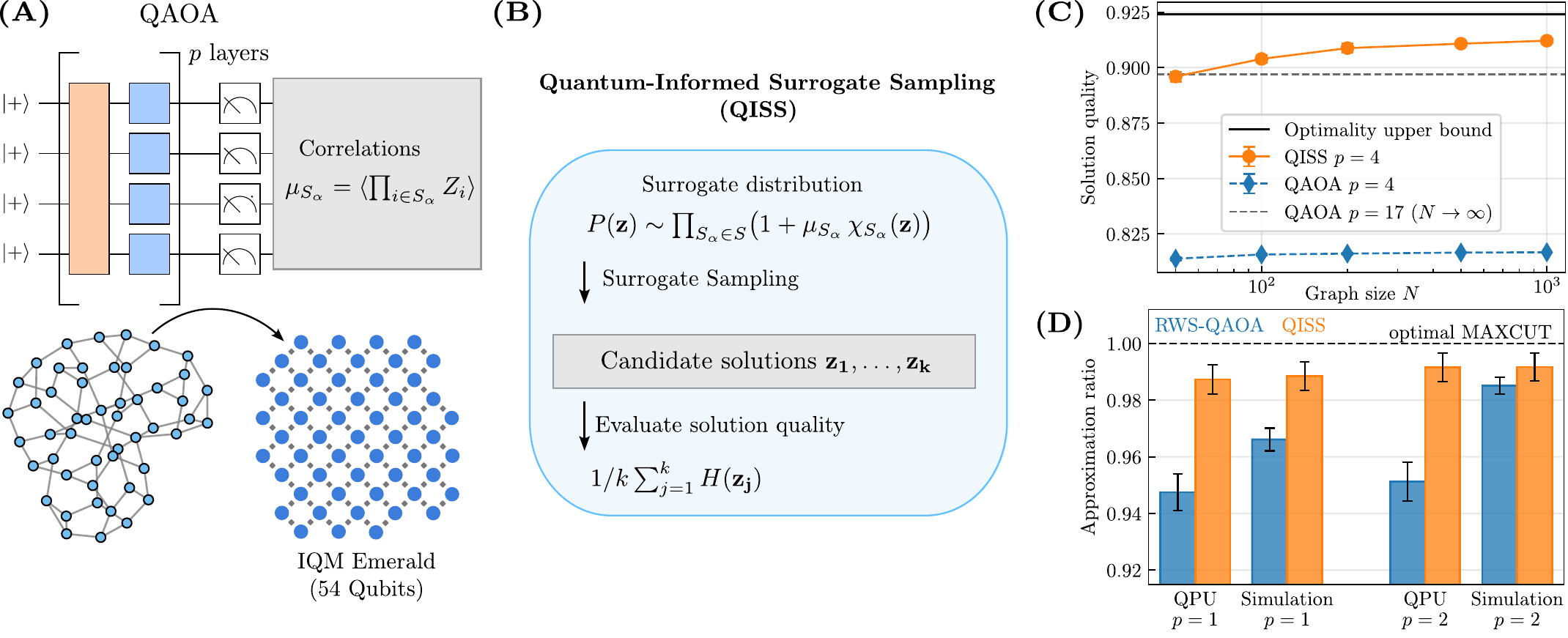}
    \caption{\textbf{Overview of the surrogate sampling procedure.} (A) We estimate the correlations $\mu_{S_\alpha}$ from a depth-$p$ QAOA circuit. The circuit can be executed on a quantum device, such as the 54-qubit IQM Emerald QPU with square-grid connectivity. (B) From a polynomially sized set of correlations, we construct a factor distribution $P(\mathbf{z})$ as a classical surrogate. We sample candidate solutions from it by MCMC and evaluate their quality. We refer to this procedure as Quantum-Informed Surrogate Sampling (QISS). (C) Sampling from the classical surrogate improves the solution quality compared to vanilla QAOA. Here, we consider the MaxCut problem on 3-regular graphs and a $p=4$ QAOA circuit. We include all non-zero weight-2 quantum mean values, i.e. $S=\{ (i,j) \mid d_G(i,j) \le 2p \}$. We use concentrated QAOA angles, making the procedure entirely training-free. (D) Device results obtained from the IQM Emerald QPU, averaged over 10 instances with $N=50$, for regularized warm-start QAOA~\cite{he2026regularizedwarmstartedquantumapproximate} at depths $p=1,2$. Blue bars show the solution quality obtained directly from the RWS-QAOA correlators; orange bars show the quality of the QISS samples generated from the same correlators. We compare the raw QPU correlators against noiseless simulation. Error bars denote the standard error of the mean.}
    \label{fig:idea_sketch}
\end{figure*}

\subsection{Summary}
In this work, we propose a scalable Quantum-Informed Surrogate Sampling (QISS) framework that generates candidate solutions from low-order quantum correlations. The main idea of QISS is sketched in Fig.~\ref{fig:idea_sketch}.
Concretely, we run a QAOA circuit at low depth $p$ on $N$ variables with parameters
$\boldsymbol{\gamma}=(\gamma_1,\dots,\gamma_p)$ and
$\boldsymbol{\beta}=(\beta_1,\dots,\beta_p)$ and estimate a selected set $\{   \mu_{S_\alpha}| S_{\alpha} \in S \}$ of low-weight expectation values of the resulting QAOA state $\ket{\psi_p(\bm{\gamma},\bm{\beta})}$, see Fig.~\ref{fig:idea_sketch}{(A)},
\begin{equation}
    \mu_{S_\alpha}
    =
    \ev{\prod_{i \in S_\alpha} Z_i}{\psi_p(\bm{\gamma},\bm{\beta})}.
    \end{equation}
In practice we only use weights $|S_{\alpha}|=1,2$ for all $S_{\alpha}$ considered. These $\mu_{S_\alpha}$ can be estimated on a quantum device by repeated state preparation and measurement and provide accessible partial information about the QAOA state. 
We then initialize a classical factor model as a surrogate distribution over the solution space $\{-1,+1\}^N$, see Fig.~\ref{fig:idea_sketch}{(B)}
\begin{equation}
    \label{eq:factors}
    P(\bm{z})
   \propto
    \prod_{S_\alpha \in S}
    \bigl( 1 + \mu_{S_\alpha} \, \chi_{S_\alpha}(\bm{z}) \bigr),
\end{equation}
with $\chi_{S_\alpha}(\bm{z}) = \prod_{i \in S_\alpha} z_i$ denoting the parity function and $S$ the collection of factor supports. We draw candidate solutions from this surrogate by Markov-chain Monte Carlo. Each factor involves only the spins in $S_{\alpha}$, so the change in $ P(\bm{z})$ under a spin flip is computed in time proportional to the number of factors containing that spin, making individual updates efficient.
As a post-processing layer, this step is modular and in principle cost-function agnostic.

The purpose of the surrogate distribution is not to reproduce, or even to approximate, the full output distribution of the quantum state. In fact, efficient classical sampling from that distribution would collapse the polynomial hierarchy~\cite{Farhi2016}. Instead, QISS serves as a practical tool for improving optimization performance and scaling to larger problems on quantum computers that are of limited size and subjected to noise.

We find that samples drawn from the surrogate distribution can significantly outperform vanilla QAOA. Solutions produced by QISS for the MaxCut problem on 3-regular graphs from QAOA correlations at depth $p=4$ outperform vanilla QAOA at depth $p=17$ on average~\cite{Farhi2025}, see Fig.~\ref{fig:idea_sketch}{(C)}. In the following, we will also show that our method is competitive with state-of-the-art classical heuristics. 

In addition, we ran the Regularized Warm-Start QAOA (RWS-QAOA) circuits of Ref.~\cite{he2026regularizedwarmstartedquantumapproximate} for MaxCut on 3-regular graphs on the 54-qubit IQM Emerald Quantum Processing Unit (QPU), and used the measured one- and two-point correlators as input to QISS. Fig.~\ref{fig:idea_sketch}{(D)} compares the approximation ratio obtained from the RWS-QAOA correlators alone with that obtained after applying QISS, both for correlators measured on the QPU and from noiseless simulation. Notably, the QISS output seems insensitive to device noise in the considered regime: the approximation ratios obtained from raw QPU correlators and from noiseless correlators are nearly indistinguishable. The noisy correlators already carry enough structure to reach the same near-optimal cuts. These results show that QISS is an effective post-processing strategy.

\subsection{Previous work}

Efforts to improve QAOA-based optimization typically act at one of two ends of the pipeline: either on the output side, by post-processing the statistics or correlations of the measured distribution in a single pass or iteratively, or on the input side, by modifying the ansatz itself through the choice of initial state, mixer, or problem encoding. Our approach is `output-side' but can be combined with `input-side' improvements.

Several output-side approaches aim to extract information from the QAOA distribution beyond the mean or the single best sampled bitstring. The Conditional Value at Risk (CVaR) objective replaces the usual expectation value of the cost function $\ev{H_C}$ by the conditional mean over the lowest $\alpha$-fraction of sampled energies, favoring rare but exceptionally good samples over a better average~\cite{Barkoutsos2020}. 
Quantum-enhanced Markov Chain Monte Carlo uses the QAOA circuit in the proposal step of the Markov chain while keeping the classical Metropolis accept/reject step so that the chain can mix polynomially faster than under local classical proposals while still provably converging to the target Gibbs distribution. ~\cite{Layden2023,Nakano2024,Marshall2026,Kawamata2026}. A related approach uses QAOA samples directly as warm starts for classical heuristics, yielding run-time gains over the classical baselines~\cite{cepaite2025}.
Other approaches use low-order statistics in a hybrid quantum-classical workflow to modify the instance itself or to guide a classical outer loop~\cite{Bravyi2019,Dupont2023,Brady2024,Finzgar2024, Wybo2026}. A prominent example is recursive QAOA~\cite{Bravyi2019}, where one uses low-depth QAOA estimates of two-point observables such as $\langle Z_i Z_j\rangle$ to identify strongly correlated pairs of variables to recursively reduce the problem. In these approaches, the quantum step is used as a correlation oracle rather than as a direct solver which is similar to our method. However, unlike these approaches our method does not rely on an outer loop and is strictly sequential.
Quantum relax-and-round follows a related philosophy, but instead injects measured quantum correlations into a classical relaxation-and-rounding pipeline, so that the eventual rounding step exploits a correlation structure informed by the variational quantum state rather than by a purely classical semidefinite (SDP) or linear (LP) programming relaxation~\cite{Dupont2023}. A related scheme rounds samples drawn from multivariate Gaussians shaped by the quantum correlations~\cite{Martinez2026}. 

Input-side approaches instead modify the QAOA ansatz itself. Warm-start methods encode classical information into the initial state and, in the strongest versions, into the mixer itself. In the framework of Ref.~\cite{Egger2021} a classical solution $\bm{z}^\star\in \{-1,+1\}^N$ defines a biased product state $\bigotimes_i \left(\sqrt{1-z_i^\star}\ket{0}+\sqrt{z_i^\star}\ket{1}\right)$, typically with a small regularization to avoid frozen dynamics, thereby importing part of the classical approximation guarantees into the ansatz rather than starting from the problem-agnostic uniform superposition. Follow-up work made this concrete for MaxCut using SDP-based warm starts, showing that custom mixers aligned with the warm-start state can lead to substantial improvements over vanilla QAOA~\cite{Tate2023,Tate2023a,He2023,Okada2024,he2026regularizedwarmstartedquantumapproximate}. A further line of research reformulates the problem into a more compact or structurally better encoding~\cite{Chancellor2019,Sawaya2023,Bako2025,Sciorilli2025,Ma2026}, exploiting that qubit count, the locality of the cost function, penalty overhead, and availability of constraint-preserving mixers are often determined by the chosen variable encoding. Therefore, logical variables may be represented implicitly through structured correlations rather than direct bit assignments. The Pauli-correlation encoding of Ref.~\cite{Sciorilli2025} represents the problem through Pauli correlations over a polynomially reduced number of qubits. Our sampling framework could potentially decode such an encoding by sampling logical solutions from the measured correlation structure.

\subsection{Organization}
This paper is organized as follows. In Sec.~\ref{sec:background}, we review QAOA and introduce the problem classes we will consider. In Sec.~\ref{sec:factor_distributions}, we introduce the surrogate model based on the mean values extracted from QAOA and the sampling procedure. In Sec.~\ref{sec:results} we present the results. We conclude in Sec.~\ref{sec:concl}

\section{Background} \label{sec:background}
\subsection{Quantum Approximate Optimization Algorithm.}
The Quantum Approximate Optimization Algorithm (QAOA)~\cite{Farhi2014}
is a variational hybrid quantum--classical algorithm designed to address
combinatorial optimization problems that can be expressed as the minimization
of an Ising-type cost Hamiltonian. A common case is a quadratic Hamiltonian on a graph $G(E,V)$ with $|V|=N$
\begin{equation}
H_C = \sum_{ij\in E} J_{ij} Z_i Z_j + \sum_{i \in V} h_i Z_i,
\end{equation}
where $Z_i$ denotes the Pauli-$Z$ operator acting on qubit $i$, and the coefficients $J_{ij}$ and $h_i$ define the problem instance on $G$.
A complementary mixing Hamiltonian
\begin{equation}
H_M = \sum_i X_i,
\end{equation}
where $X_i$ is the Pauli-$X$ operator, generates transitions between the computational basis states and ensures ergodic exploration of the solution space.

Starting from an initial product state, typically $\ket{+}^{\otimes N}$, QAOA alternates between unitaries generated by $H_C$ and $H_M$ to prepare the parameterized state
\begin{equation} \label{eq:qaoa}
|\psi_p(\boldsymbol{\gamma},\boldsymbol{\beta})\rangle =
\prod_{k=1}^{p} e^{-i \beta_k H_M} e^{-i \gamma_k H_C} |\!+\!\rangle^{\otimes N},
\end{equation}
where the parameters
$\boldsymbol{\gamma}=(\gamma_1,\dots,\gamma_p)$ and
$\boldsymbol{\beta}=(\beta_1,\dots,\beta_p)$
are optimized classically to minimize the energy expectation
$\langle H_C \rangle = \langle \psi_p(\boldsymbol{\gamma},\boldsymbol{\beta}) | H_C | \psi_p(\boldsymbol{\gamma},\boldsymbol{\beta}) \rangle$.

The circuit depth $p$ controls the expressiveness of the ansatz. For small $p$, QAOA yields shallow circuits compatible with Noisy Intermediate-Scale Quantum (NISQ) devices~\cite{Preskill2018}, with performance governed by local graph structure. As $p$ increases, correlations propagate across progressively larger regions of the graph, and in the limit of $p\to\infty$ the ansatz can approximate a digitized adiabatic evolution toward the ground state. 

\subsection{Problems: MaxCut and Maximum Independent Set} \label{ssec:models}


To benchmark the sampling framework, we consider two standard Ising optimization problems: Maximum Cut (MaxCut) and Maximum Independent Set (MIS). However, many more NP-hard graph optimization problems can be cast directly into an Ising formulation with binary spins $z_i \in \{-1,+1\}^N$ for which QAOA can be directly applied~\cite{Lucas2014}. 

For the MaxCut problem on a graph $G=(V,E)$, one seeks to partition the vertices into two sets such that the number of edges connecting them is maximized. The corresponding Ising Hamiltonian is
\begin{equation}\label{eq:H_MC}
H_{\mathrm{MAXCUT}} = -\frac{1}{2}\sum_{(i,j)\in E} (1 - Z_i Z_j),
\end{equation}
and the ground state encodes the optimal cut. As a performance metric we will use the cut fraction $-\ev{H_{\mathrm{MAXCUT}}}/|E|$ which quantifies the fraction of edges that are cut. 

On the other hand, the MIS problem aims to find the largest subset of vertices with no adjacent pairs. Defining $z_i = +1$ for included vertices in the independent set and $z_i = -1$ otherwise, one can write
\begin{equation}\label{eq:H_MIS}
H_{\mathrm{MIS}}^{\lambda} = -\sum_i Z_i + \lambda \sum_{(i,j)\in E} \frac{1+Z_i}{2}\frac{1+Z_j}{2},
\end{equation}
where the penalty term proportional to $\lambda\geq 1$ enforces the independence constraint. As a performance metric, we will use the independence ratio which quantifies the fraction of vertices included in the independent set, i.e. $\ev{H_{\mathrm{MIS}}^{\lambda=1}}/N$.

We will consider both problems on random $3$-regular graphs. A key feature of QAOA is that, at fixed depth $p$, it is a \emph{local} algorithm for sparse structures like random regular graphs. Indeed, each application of the cost unitary $e^{-i\gamma_k H_C}$ enlarges the support of a local observable in the Heisenberg picture only along edges of the graph with couplings $J_{ij}\neq 0$, while the mixer $e^{-i\beta_k H_M}$ acts on individual qubits. As a result, a local observable spreads only within a finite depth-$p$ light cone, so the expectation value of any local term in $H_C$ depends only on its radius-$p$ neighborhood. For example, the expectation value of an edge term $Z_i Z_j$ is completely determined by the part of the problem instance lying within graph distance $p$ of the edge $(i,j)$. Couplings outside this light cone do not affect its value. This locality is central to the analytical tractability of low-depth QAOA on sparse $d$-regular graphs~\cite{Basso2021}. On such graphs with sufficiently large girth, the light cones are cycle-free with high probability~\cite{Makover2006} and therefore coincide with a finite rooted $d$-regular tree when $N\rightarrow \infty$. Consequently in this limit, all edges have the same local QAOA neighborhood, and the lowest expected energy density and corresponding optimal QAOA angles can be evaluated by analyzing these tree structures rather than the full graph. We therefore fix the QAOA angles in this way throughout the manuscript; the angles for both problems can be found in Ref.~\cite{Wybo2025}. The resulting tree structure is also well suited to tensor-network methods, since tensor networks on trees can be contracted efficiently~\cite{Shi2008}.

\section{Quantum-informed surrogate sampling} \label{sec:factor_distributions}

\begin{figure*}
    \centering
    \includegraphics[width=\textwidth]{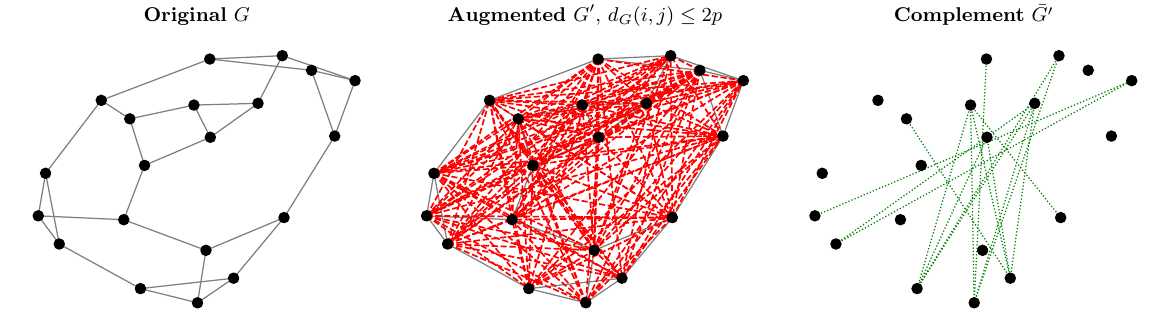}
    \caption{Starting from an original graph $G$, we construct an augmented graph $G'$. There is an edge in the augemented graph if $d_G(i,j) \leq 2p$, in this figure we consider $p=2$. We also show the complement graph of $G'$ to highlight the distinction from the complete graph. }
    \label{fig:sketch}
\end{figure*}

The purpose of Quantum-Informed Surrogate Sampling (QISS) is to construct candidate solutions $\bm{z} \in \{-1,+1\}^N$ to the combinatorial optimization problems introduced in the previous section~\ref{ssec:models} from low-weight quantum correlators, in particular those measured on a QAOA state~\eqref{eq:qaoa}. We focus on expectation values of Pauli-$Z$ operators supported on a small subset of qubits $S_\alpha \subseteq \{1,\dots,N\}$ with $|S_\alpha|\in O(1)$,
\begin{equation}
    \mu_{S_\alpha}
    =
    \ev{\prod_{i \in S_\alpha} Z_i}{\psi_p(\bm{\gamma}^\star,\bm{\beta}^\star)}
    \in [-1,+1].
\end{equation}

Using these mean values, we define a factorized probability distribution over the domain
$\bm{z} \in \{-1,+1\}^N$ and the collection $S$ over \textit{factor supports} $S_{\alpha}$ as
\begin{equation}
    \label{eq:factors}
    P(\bm{z})
    =
    \frac{1}{\mathcal{Z}}
    \prod_{S_\alpha \in S}
    \bigl( 1 + \mu_{S_\alpha} \, \chi_{S_\alpha}(\bm{z}) \bigr),
\end{equation}
where $\chi_{S_\alpha}(\bm{z}) = \prod_{i \in S_\alpha} z_i$ denotes the parity (character) function associated with the subset $S_\alpha$, and the normalization constant, or partition function, is given by
\begin{equation}
    \mathcal{Z}
    =
    \sum_{\bm{z}}
    \prod_{S_\alpha \in S}
    \bigl( 1 + \mu_{S_\alpha} \, \chi_{S_\alpha}(\bm{z}) \bigr).
\end{equation}
Computing $\mathcal{Z}$ exactly is generally intractable, as it involves a sum over exponentially many configurations. Nevertheless, this factor distribution has the advantage that efficient Markov-chain Monte Carlo (MCMC) sampling schemes can be constructed, as discussed below.

We first note that the distribution in Eq.~\eqref{eq:factors} is equivalent to a Gibbs distribution for an effective Ising Hamiltonian $H'$,
\begin{equation} \label{eq:gibbs_factors}
    P(\bm{z})
    =
    \frac{1}{\tilde{\mathcal{Z}}_T}
    \exp\!\left(
        -\frac{1}{T}
        \sum_{S_\alpha \in S}
        J_{S_\alpha} \, \chi_{S_\alpha}(\bm{z})
    \right)
    =
    \frac{1}{\tilde{\mathcal{Z}}_T}
    \exp\!\left( -\frac{1}{T} H'(\bm{z}) \right),
\end{equation}
with temperature $T=1$ and couplings
\begin{equation}
    J_{S_\alpha}
    =
    -\operatorname{arctanh}(\mu_{S_\alpha})
    =
    -\frac{1}{2}
    \ln\!\left(
        \frac{1+\mu_{S_\alpha}}{1-\mu_{S_\alpha}}
    \right).
\end{equation}
Thus, the couplings in $H'$ are directly determined by the measured QAOA expectation values. Importantly, this effective model lives on a graphical structure that is defined by the collection of factor supports $S$, which does not need to correspond to the original problem graph.

The number of terms in $H'$ is given by $|S|$. In this work, we assume quadratic Hamiltonians and will consider two natural choices for $S$. The first one simply corresponds to the structure of the original problem graph $G=(V,E)$, i.e.
\begin{equation}
    S=E\cup V.
\end{equation}

The second choice augments the graph $G(E,V)$ to $G'(E',V)$ according to the light-cone structure of the QAOA circuit. We consider $E'=\{(i,j)\mid d_G(i,j)\le 2p\}$ with $d_G(i,j)$ the graph distance, or the shortest-path length, between vertices $i$ and $j$ in $G$; $p$ is the QAOA depth. This is exactly the set of pairs whose reverse light cones overlap and whose correlator therefore does not factorize into one-body terms
\begin{equation} \label{eq:S_non_edge}
    S'=E' \cup V.
\end{equation}

These two structures can differ substantially depending on $p$. For example, if $G$ is a random regular graph and thus locally tree-like, the augmented graph $G'$ induced by $E'$ is far from locally tree-like: by design the inclusion of all edges within distance $2p$ creates dense local clusters and short cycles. This is illustrated in Fig.~\ref{fig:sketch} where we start from a 3-regular graph on with 20 nodes and augment the structure with edges representing nontrivial correlators at $p=2$.

If $ d_G(i,j) \le 2p $ for all pairs, then $G'\equiv K_N$, the complete graph on $N$ nodes, and the full correlation matrix
\begin{equation}
    C_{ij}
    =
    \ev{Z_i Z_j}{\psi_p(\bm{\gamma^{\star}},\bm{\beta^{\star}})}
\end{equation}
is required.
For a 3-regular graph, this regime is reached when
\begin{equation}
    p \approx \log_2(N/6)
    \qquad
    \Rightarrow N \approx 6 \cdot 2^p.
\end{equation}
Since the light-cone neighborhood grows only as $O(2^p)$, its size is set by the depth rather than by $N$. The correlators $C_{ij}$ can therefore be computed exactly irrespective of $N$, in practice up to $p\approx 5$ with classical tensor-network techniques~\cite{Gray2018quimb}.

\subsection{Sampling from the factor distribution}
\label{sec:factor_sampling}

We sample from the factor distribution~\eqref{eq:factors} by constructing a Markov chain based on single-site conditional updates. The key observation is that, although the full probability distribution in Eq.~\eqref{eq:factors} involves the partition function $\mathcal{Z}$, the conditional probabilities required for the Markov chain do not, and can thus be computed efficiently.

Consider a configuration of all variables except $i$, denoted by $\bm z_{-i}$. The conditional probability that $z_i$ takes a given value depends exclusively on the subset of factors that involve $i$. We denote this subcollection by ${T_i = \{ S_\alpha \in S \mid i \in S_\alpha \}}$. Defining the unnormalized weights as
\begin{equation}
    w(z_i)
    =
    \prod_{S_\alpha \in T_i}
    \bigl(1 + \mu_{S_\alpha} \, \chi_{S_\alpha}(\bm z)\bigr),
\end{equation}
the conditional probability takes the form
\begin{equation}
    P(z_i= + 1 \mid \bm z_{- i})
    =
    \frac{w(z_i=+1)}{w(z_i=+1)+w(z_i=-1)}.
\end{equation}
Since $\chi_{S_\alpha}(\bm z)=z_i\prod_{j\in S_\alpha\setminus\{i\}}z_j$, the weights
can be evaluated explicitly as
\begin{align}
    w(z_i=+1)
    &=
    \prod_{S_\alpha \in T_i}
    \Bigl(1+\mu_{S_\alpha}\!\!\prod_{\substack{j\in S_\alpha\\ j\neq i}}\!z_j\Bigr),\\
    w(z_i=-1)
    &=
    \prod_{S_\alpha \in T_i}
    \Bigl(1-\mu_{S_\alpha}\!\!\prod_{\substack{j\in S_\alpha\\ j\neq i}}\!z_j\Bigr).
\end{align}
The update probability can be expressed in compact form,
\begin{equation}
    P(z_i=+1 \mid \bm z_{- i})
    =
    \sigma\!\left(
        \log w(z_i=+1) - \log w(z_i=-1)
    \right),
    \qquad
    \sigma(x)=\frac{1}{1+e^{-x}}.
\end{equation}

Then, starting from an initial configuration $\bm z \in \{-1,+1\}^N$, one sweep of the Markov chain consists of updating each variable once in random order. We perform a sweep as follows, for each site $i=1,\dots,N$: (i) Compute the conditional probability $ p_i = P(z_i=+1 \mid \bm z_{-i})$ induced by the factor distribution $P(\bm z)$ in Eq.~\eqref{eq:factors}; (ii) Draw $r_i \sim \mathrm{Bernoulli}(p_i)$; (iii) Set $z_i = +1$ if $r_i=1$ and $z_i=-1$ otherwise; (iv) Update all factors $\chi_{S_\alpha}$ with $S_\alpha \in T_i$ that depend on $z_i$.

This procedure defines a single-site MCMC sampler. By construction, the transition kernel satisfies detailed balance with respect to $P(\bm z)$ (see Appendix~\ref{app:MCMC_proofs}), and the normalization constant of $P$ is never required, as only local conditional probabilities are evaluated. Since the state space is finite and all configurations are reachable through successive single-spin updates, the chain is irreducible and aperiodic (see Appendix~\ref{app:MCMC_proofs}). Hence $P(\bm z)$ is the unique stationary distribution of the Markov chain, and the sampling cost per sweep is $\Theta(|S|)$.

In our simulations, the Markov chain is initialized either from a uniformly random configuration or from a configuration aligned with the signs of the measured single-site observables, $z_i=\mathrm{sign}(\ev{Z_i}{\bm{\gamma^{\star}},\bm{\beta^{\star}}})$, when available. We perform $300$ burn-in sweeps to allow convergence toward stationarity. Subsequently, $5000$ additional sweeps are carried out, and every tenth configuration is retained in order to reduce autocorrelations. This results in approximately $500$ effectively independent samples from $P(\bm z)$ for our simulations.

\subsection{Mean values of the factor distribution} \label{ssec:factor-means}

We now analyze the expectation values $\ev{\chi_T}_P$ with respect to the factor distribution~\eqref{eq:factors}, which are needed to estimate the average performance of the surrogate sampling. In general, these expectation values have no closed form in terms of $\mu_{S_\alpha}$ without additional assumptions on the graphical structure of the factor support collection $S$. To make this explicit, we expand the product over factors. Let $A\subseteq S$ denote a subcollection of factor supports, and define its \emph{parity support} as
\begin{equation}
    \Pi(A)
    =
    \bigoplus_{S_\alpha\in A} S_\alpha ,
\end{equation}
where $\bigoplus$ denotes the symmetric difference of sets, such that $\Pi(A)$ contains precisely those indices that appear in an odd number of the sets contained in $A$. We then have
\begin{equation}
    W(\bm{z})
    =
    \prod_{S_\alpha\in S}
    \bigl(1+\mu_{S_\alpha}\chi_{S_\alpha}(\bm z)\bigr)
    =
    \sum_{A\subseteq S}
    \Bigl(
        \prod_{S_\alpha\in A}\mu_{S_\alpha}
    \Bigr)
    \chi_{\Pi(A)}(\bm z).
\end{equation}
Using the orthogonality of the parity characters
\begin{equation}
    \frac{1}{2^N}
    \sum_{\bm z}
    \chi_{T_1}(\bm z)\chi_{T_2}(\bm z)
    =
    \delta_{T_1,T_2},
\end{equation}
the expectation value of $\chi_T$ can be written as
\begin{equation} \label{eq:ev_chi}
    \ev{\chi_T}_P
    =
    \frac{ \sum_{\bm z} \chi_T(\bm{z}) W(\bm{z}) }{\sum_{\bm z}  W(\bm{z})}
    =
    \frac{
        \displaystyle
        \sum_{\substack{A\subseteq S\\ \Pi(A)=T}}
        \prod_{S_\alpha\in A}\mu_{S_\alpha}
    }{
        \displaystyle
        \sum_{\substack{A\subseteq S\\ \Pi(A)=\emptyset}}
        \prod_{S_\alpha\in A}\mu_{S_\alpha}
    } .
\end{equation}
This expression shows why closed-form evaluation is difficult in general. The numerator receives contributions from all subsets of factors whose supports combine to $T$ under symmetric difference, while the denominator contains all combinations of subsets whose supports cancel to the empty set. The number of such subsets can grow exponentially with the number of factors.

The exact identity also elucidates the perturbative weak coupling regime. If $T=S_\alpha\in S$, then the single-factor subset $A=\{S_\alpha\}$ contributes $\mu_{S_\alpha}$. All other contributions arise from products of two or more factor strengths whose supports combine to the same parity, i.e.
\begin{equation}
    \ev{\chi_{S_\alpha}}_P
    =
    \mu_{S_\alpha}
    +
    \text{higher-order parity corrections}.
\end{equation}
This approximation is controlled only when these higher-order corrections are small.

A simple exact case is obtained when the factor collection contains only single-variable factors. Then
\begin{equation}
    P(\bm z)
    \propto
    \prod_i (1+\mu_i z_i),
\end{equation}
so the variables are independent and $\ev{z_i}_P=\mu_i$. 
In particular, if $\mu_i=1$, the variable is frozen to $z_i=+1$, while if $\mu_i=-1$, it is frozen to $z_i=-1$.

Exact computation is also possible when the graphical structure underlying the factor distribution is sufficiently simple. The graphical structure can be represented by a \emph{factor graph}: a bipartite graph with one node per variable and one node per factor, where an edge joins variable $i$ to factor $\chi_{S_{\alpha}}$ whenever $i\in S_{\alpha}$, and never two variables or two factors.
If such a factor graph induced by the variables and factors corresponding to $S$ is a tree, belief-propagation can compute the partition function and all marginals exactly, in a single pass with a cost that is linear in the number of variables~\cite{Pearl1988}. This is directly relevant to the case where $S= E \cup V $considered in this work. When every factor couples exactly two variables, the factor graph reduces (up to the factor nodes on each edge) to the ordinary interaction graph. If that graph is a random regular graph, it is locally tree-like, so any neighborhood of bounded radius is acyclic with high probability. Belief propagation is therefore not fully exact, but it is asymptotically exact for local quantities as $N \to \infty$.

In this work, we estimate the required mean values numerically by sampling from the corresponding Gibbs distribution using the MCMC method described above. To ensure statistically reliable estimates, it is necessary to account for temporal correlations in the Markov chain. Given a time series $O_1,\dots,O_n$ of measured observables, we estimate the autocovariance function using the convolution theorem,
\begin{equation}
\widehat{C}_t
=
\mathcal{F}^{-1}\!\bigl(|\mathcal{F}(O-\bar O)|^2\bigr)_t,
\qquad
t=0,\dots,n-1,
\end{equation}
where $\bar O$ denotes the sample mean and $\mathcal{F}$ the discrete Fourier transform. Normalizing by $\widehat{C}_0$ yields the autocorrelation function $\widehat{\rho}_t = \mathbb{E}[(O_k - \bar{O})(O_{k+t} - \bar{O})]/\sigma^2$ with $\sigma^2 = \mathbb{E}[(O_t-\bar{O})^2]$. The integrated autocorrelation time is estimated using the initial positive sequence rule~\cite{Geyer1992},
\begin{equation}
\widehat{\tau}_{\mathrm{int}}
=
\frac{1}{2}
+
\sum_{k\ge 1}
\bigl(\widehat{\rho}_{2k-1}+\widehat{\rho}_{2k}\bigr),
\end{equation}
with the sum truncated at the first $k$ for which
$\widehat{\rho}_{2k-1}+\widehat{\rho}_{2k}\le 0$. This yields a robust estimate of
the effective number of independent samples,
\begin{equation}
\widehat{n}
\approx
\frac{n}{2\,\widehat{\tau}_{\mathrm{int}}},
\end{equation}
which we use to quantify statistical uncertainty in all reported averages, i.e. $\mathrm{Var}(\bar{O})= \sigma^2 / \widehat{n}$.

\subsection{Factor distribution representation} 
In this section, we take the reverse point of view. Starting from the QAOA output distribution 
\(Q(\bm{z})=|\braket{\bm{z}}{\psi(\bm{\gamma}^\star,\bm{\beta}^\star)}|^2\), or more generally from any distribution over bitstrings $Q(\bm{z})$, we identify the approximations that lead to the factor distribution used in our surrogate model.
This can be seen by writing out the moments as
\begin{equation}
    \mu_{S_\alpha}  = \sum_{\bm{z}} Q(\bm{z}) \chi_{S_\alpha}(\bm{z}),
\end{equation}
and applying the inverse Hadamard-Walsh transform
\begin{equation}
   Q(\bm{z}) = \frac{1}{2^N} \sum_{S_\alpha \in S}  \mu_{S_\alpha}  \chi_{S_\alpha}(\bm{z}).
\end{equation}
Here, the sum runs over all $2^N=|S|$ subsets \(S_\alpha \subseteq [N]\). Equivalently, each subset $S_\alpha$ can be identified with a string \(\bm{s}_\alpha\in\{-1,+1\}^N\), where $s_{\alpha i}=+1$ iff \(i\in S_\alpha\).
Considering the logarithm of the distribution, assuming $Q(\bm{z})$ has full support,
\begin{equation}
   \log(Q(\bm{z})) =  \sum_{S_\alpha}  L_{S_\alpha}  \chi_{S_\alpha}(\bm{z}),
\end{equation}
with coefficients $L_{S_\alpha} = \frac{1}{2^N} \sum_{\bm{z}}   \log(Q(\bm{z})) \chi_{S_\alpha}(\bm{z})$, implies $Q(\bm{z}) = \exp(\sum_{S_\alpha} L_{S_\alpha}  \chi_{S_\alpha}(\bm{z}))$. By the nature of the characters $\chi_{S_\alpha} \in\{-1,+1\}$ we can write 
\begin{equation}
    \exp(L_{S_\alpha}  \chi_{S_\alpha}(\bm{z})) = \cosh(L_{S_\alpha}) +\chi_{S_\alpha}(\bm{z}) \sinh(L_{S_\alpha}) = \cosh L_{S_\alpha} \left[ 1 + \chi_{S_\alpha}(\bm{z}) \tanh(L_{S_\alpha}) \right].
\end{equation}
Hence, we can recover the factor form of $Q(\bm{z})$ 
\begin{equation}
   Q(\bm{z}) \propto \prod_{S_\alpha} (1 + \tanh(L_{S_\alpha}) \chi_{S_\alpha}(\bm{z}) ),
\end{equation}
where we dropped the normalization. 
From this it can be seen that our surrogate model~\eqref{eq:factors} is obtained by making the following (strong) approximations: (i) Restricting $|S|$ such that it only contains (a polynomially large subset of) low-weight observables; (ii) Taking $\tanh(L_{S_\alpha}) \approx \mu_{S_\alpha}$.

\subsection{Consequences of MCMC sampling and heuristic maximum-entropy structure of the surrogate distribution}

\begin{figure}
    \centering
    \includegraphics[width=\textwidth]{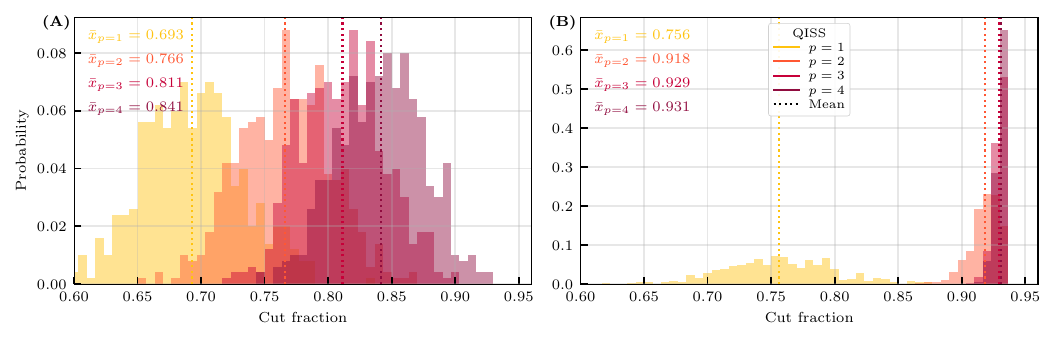}
    \caption{\textbf{Cut fractions obtained with QISS} for a single MaxCut instance on a 3-regular graph of size $N=100$. The histograms show the cut fractions corresponding to 500 samples generated via QISS based on the mean value input. (A) Only the QAOA edge correlators $E$ are included in the sampling model. (B) All non-zero QAOA edge correlators $E'$ are included in the sampling model. For this particular instance, the optimal cut fraction is $14/15\; (= 0.933...$). }
    \label{fig:maxcut_histograms}
\end{figure}

In this section, we interpret the surrogate factor model through its relation to maximum-entropy modeling, and explain in what sense it departs from it. We consider first an idealized limit in which the surrogate concentrates on optimal configurations, then the realistic setting in which the local mapping $J_{S_\alpha}=\operatorname{arctanh}(\mu_{S_\alpha})$ replaces exact moment matching, and finally the practical consequences for MCMC sampling and noise robustness.

Consider first an idealized limit in which the moments defining the factor distribution are those of the optimal-solution distribution itself. This arises, for instance, when QAOA outputs a uniform superposition $Q_\Omega$ over the optimal set $\Omega$ (possibly only in the $p\to\infty$ limit), so that $\mu_S = \ev{\chi_S}_{Q_\Omega}$. A moment reaches $\mu_S = \pm 1$ precisely when the parity $\chi_S$ takes the same value on every optimal configuration. Its factor $1+\mu_S\chi_S(\bm z)$ then acts as a hard constraint, vanishing on all configurations with $\chi_S(\bm z) = -\mu_S$, while the remaining moments ($|\mu_S|<1$) only reweight configurations without excluding any. If the deterministic moments collectively characterize $\Omega$, the resulting distribution is supported exactly on the optimal set. Sampling then has a one-sided guarantee: starting from any optimal configuration, every configuration the chain can reach is itself optimal, so the surrogate can never turn an optimal input into a suboptimal one. This does not imply that the chain mixes efficiently. If different optima are separated by configurations of zero probability, then local MCMC updates can become frozen. The chain then simply remains at its starting optimum: it fails to explore the rest of $\Omega$, but it never worsens the solution it was given. 

Away from this idealized limit, the surrogate defines a moment-informed heuristic Hamiltonian $H'$ whose couplings follow from the measured quantum correlators through the local analytic mapping $J_{S_\alpha}=\operatorname{arctanh}(\mu_{S_\alpha})$, see Eq.~\eqref{eq:gibbs_factors}.

Because each coupling $J_{S_\alpha}$ is fixed from its own moment alone, via the relation that is only exact for a single factor, the induced Gibbs distribution does not in general reproduce its own defining moments, $\ev{\chi_{S_\alpha}}_P \neq \mu_{S_\alpha}$. Equality could only hold in the special cases identified in Sec.~\ref{ssec:factor-means} under Eq.~\eqref{eq:ev_chi}: a single factor, the weak-coupling regime or a tree-structured factor graph. The construction should thus be read as a structured classical surrogate whose low-order parameters are set by the quantum correlations, while its higher-order correlations are fixed implicitly by the induced interaction structure $G'$.

The exponential form
\begin{equation}
Q(\bm z)
\propto
\exp \left(
\sum_{S_\alpha\in S}
J_{S_\alpha}\chi_{S_\alpha}(\bm z)
\right)
\end{equation}
nonetheless carries an information-theoretic reading. If the couplings were chosen to enforce the moment constraints $\sum_{\bm z} Q(\bm z)\chi_{S_\alpha}(\bm z) = \mu_{S_\alpha}$ for all $S_\alpha \in S$, then $Q$ would be the unique distribution maximizing the Shannon entropy $-\sum_{\bm z} Q(\bm z)\log Q(\bm z)$ subject to those constraints. Such a model would, however, inherit the QAOA correlations it was built from rather than surpass them. The surrogate instead forgoes the global moment-matching problem and adopts the explicit local approximation $J_{S_\alpha} = \operatorname{arctanh}(\mu_{S_\alpha})$, which is exact for an isolated factor but only approximate for overlapping ones. The result is a maximum-entropy-\emph{inspired} model that incorporates the available quantum information without introducing additional tunable parameters or requiring a global convex optimization.

From an optimization perspective, sampling the surrogate at fixed temperature $T=1$ biases the candidate solutions toward low-energy configurations of the effective Hamiltonian $H'$, while still producing a range of configurations. Since the couplings of $H'$ are set by the quantum correlations, these configurations reflect the structure encoded in the QAOA state; the empirical finding of Sec.~\ref{sec:results} is that they also tend to be good solutions of the original optimization problem. Finally, in terms of stability, the smooth monotonic dependence $J_{S_\alpha} = \operatorname{arctanh}(\mu_{S_\alpha})$ ensures robustness with respect to noise in the estimated quantum correlations, since small perturbations in $\mu_{S_\alpha}$ induce continuous deformations of the surrogate distribution rather than abrupt structural changes in $G'$.

\section{Results} \label{sec:results}
In this section, we benchmark QISS on MaxCut and MIS, both on random 3-regular graphs, using correlators from QAOA, from RWS-QAOA, and from the 54-qubit IQM Emerald QPU. Throughout we use the fixed tree-optimal angles of Ref.~\cite{Wybo2025} for vanilla QAOA, making the procedure entirely training-free.

\begin{figure}
    \centering
    \includegraphics[width=\textwidth]{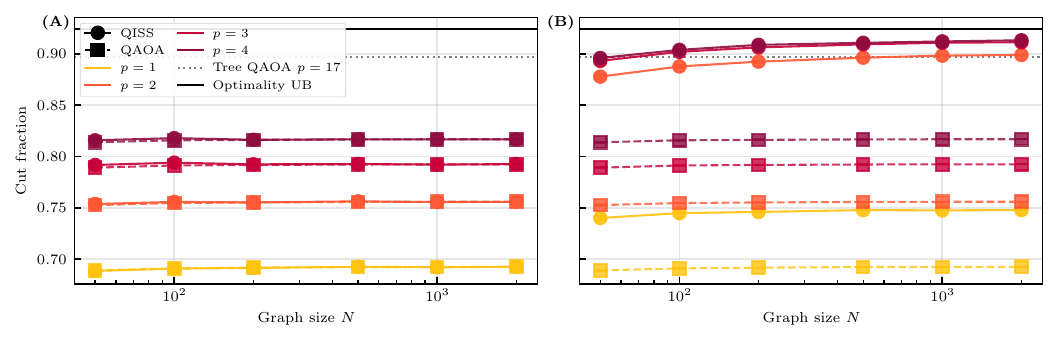}
    \caption{\textbf{QISS for MaxCut on 3-regular graphs based on QAOA correlations.} Average cut fractions from vanilla QAOA and from QISS, for various QAOA depths $p$, averaged over 40 instances per system size $N$. (A) Only the edge correlators are included in the surrogate model; QISS reproduces the QAOA cut fraction. (B) All non-zero correlators within the light cone ($d_G(i,j)\le 2p$) are included, giving a large improvement over vanilla QAOA. The dotted line marks vanilla tree QAOA at $p=17$~\cite{Farhi2025}.}
    \label{fig:maxcut_means}
\end{figure}

\subsection{MaxCut}

We first present results for MaxCut on 3-regular graphs. We sample a problem-instance set $\{G=(V,E)\}$ of 40 randomly chosen 3-regular graphs for each system size $N=|V|=50,100,200,500,1000,2000$. As noted in Sec.~\ref{ssec:models}, these graphs are locally tree-like, so the depth-$p$ correlations have bounded locality and can be computed exactly for small $p$, at a cost scaling exponentially in $p$. We evaluate them by exact tensor-network contraction~\cite{Gray2018quimb} using the tree-optimal angles $(\bm\gamma^\star,\bm\beta^\star)$ of Ref.~\cite{Wybo2025} for vanilla QAOA.

\subsubsection{QAOA correlations}

We first apply QISS to correlators from vanilla QAOA. Since the MaxCut Hamiltonian is $\mathbb{Z}_2$-symmetric, all one-body expectation values vanish, so it suffices to compute the two-body correlators $\ev{Z_i Z_j}{\bm\gamma^\star,\bm\beta^\star}$ for pairs with $d_G(i,j)\le 2p$. Based on these, we apply QISS (see Sec.~\ref{sec:factor_sampling}) to each instance, using either the edge set $S=E$ or the enlarged set $S'=E'=\{(i,j)\mid d_G(i,j)\le 2p\}$. For each sample $\bm z$, we evaluate the cut fraction $-H_{\mathrm{MAXCUT}}(\bm z)/|E|$, the normalized cost of Eq.~\eqref{eq:H_MC}. Fig.~\ref{fig:maxcut_histograms} shows histograms of these cut fractions for a single $N=100$ instance, illustrating a large shift towards near-optimal solutions by applying QISS on the factor collection $S'$.



Fig.~\ref{fig:maxcut_means} compares the average cut fractions obtained with QISS and vanilla QAOA. When only the edge correlations are included in the factor distribution, the sampler does not improve on the QAOA on average, as can be seen from Fig.~\ref{fig:maxcut_means}{(A)}. This follows from the locally tree-like structure of the factor graph, which is inherited directly from the problem graph. Hence, in this case the prescribed moments, and thus the cut fractions, are reproduced, see also Sec.~\ref{ssec:factor-means}. However, if we input the correlations corresponding to the augmented structure in Fig.~\ref{fig:sketch}, sampling from the surrogate generates solutions that substantially outperform vanilla QAOA, even when based on shallow correlators at $p\approx 2,3$. This is shown in Fig.~\ref{fig:maxcut_means}{(B)}: our method exceeds the average cut fraction of vanilla $p=17$ tree QAOA (dotted line)~\cite{Farhi2025}, the largest depth for which the tree-optimal angles have been computed to our knowledge. 

In Appendix~\ref{app:thermal}, we examine a variant of QISS in which the input correlations are drawn from a thermal state rather than from QAOA, allowing for a direct comparison between the two and indicating that QAOA correlations lead to better results.

\subsubsection{RWS-QAOA correlations}

\begin{figure}
    \centering
    \includegraphics{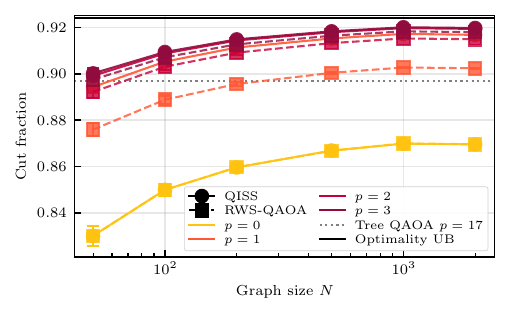}
    \caption{\textbf{QISS for MaxCut on 3-regular graphs based on RWS-QAOA correlations.} Average cut fractions obtained from the RWS-QAOA and combined with QISS postprocessing at different depths $p$. Here, $p=0$ corresponds to the expectation values estimated with respect to the initial RWS-QAOA state; the QISS postprocessing is unable to improve. For reference, the $p=17$ mean cut fraction from vanilla QAOA is shown~\cite{Farhi2025}. QISS improves upon the supplied RWS-QAOA correlations, with the relative improvement decreasing as we approach the optimality upper bound with increasing $p$.}
    \label{fig:maxcut_ws_means}
\end{figure}

Following Ref.~\cite{he2026regularizedwarmstartedquantumapproximate}, we can use \emph{regularized warm-start QAOA} (RWS-QAOA) to generate correlations that can be supplied to QISS. In RWS-QAOA, a linear-time classical preprocessing step is performed that introduces biases into the initial QAOA state. Instead of initializing each qubit as $\ket{+}$, the state $R_y(\theta_i)\ket{+}$ is prepared on each qubit, where $\theta_i$ is given by the classical preprocessing routine. Additionally, the mixing Hamiltonian is modified relative to vanilla QAOA: instead of the simple $\sum_i X_i$, the mixer is chosen to be $H_M^{\mathrm{RWS}}(\boldsymbol{\theta})=\sum_i \sin(\theta_i) X_i + \cos(\theta_i) Z_i$. The angles $\theta_i$ are determined by minimizing the following regularized relaxation of the MaxCut cost function~\cite{he2026regularizedwarmstartedquantumapproximate}
\begin{equation}
\label{eq:regularized-objective-rwsqaoa}
\boldsymbol{\theta}^{\star}
=
\operatorname*{arg\,min}_{\boldsymbol{\theta}\in[0,\pi]^N}
\left[
\frac{1}{2}
\sum_{(i,j)\in E}
\cos(\theta_i)\cos(\theta_j)
-
\eta
\sum_{i=1}^{N}
\sin^2(\theta_i)
\right],
\end{equation}
where $\eta=0.6$ is the strength of the regularization term penalizing near-bitstring states, empirically determined in Ref.~\cite{he2026regularizedwarmstartedquantumapproximate}. In this work we determine $\boldsymbol{\theta}^*$ by minimizing the relaxed objective Eq.~\eqref{eq:regularized-objective-rwsqaoa} using at most 1000 iterations of the gradient-based L-BFGS-B~\cite{Liu1989} optimizer with 1000 random initializations, to mitigate the nonconvex nature of the relaxed objective. As before, we then compute the required expectation values using exact tensor network contraction~\cite{Gray2018quimb} for depths $p\in\left\{1,2,3\right\}$. We use the RWS-QAOA parameters $\boldsymbol{\gamma},\boldsymbol{\beta}$ as provided by Ref.~\cite{he2026regularizedwarmstartedquantumapproximate}.

In Fig.~\ref{fig:maxcut_ws_means} we show the cut fractions obtained from tensor-network simulations of RWS-QAOA and after post-processing with QISS. At $p=0$ we supply only the one-body moments $\langle Z_i\rangle$ of the initial RWS-QAOA state (i.e. before any phase-separator or mixer unitary is applied), so the surrogate factorizes into independent single-site distributions with marginals $\langle Z_i\rangle$, the exact case of Sec.~\ref{ssec:factor-means}. Sampling then reproduces the mean-field cut, which coincides with the bare RWS-QAOA estimate at $p=0$, hence the $p=0$ markers in Fig.~\ref{fig:maxcut_ws_means} overlap. Once QAOA evolution generates genuine two-body correlations ($p\ge1$), QISS consistently improves upon the RWS-QAOA cut fractions across all system sizes. Notably, already at $p=1$ the QISS cut fractions exceed the vanilla tree-QAOA value at $p=17$ (dotted line)~\cite{Farhi2025}. The relative improvement over RWS-QAOA is largest at low depth and shrinks with increasing $p$. This is expected, since there is little room left to improve once RWS-QAOA is already close to optimal.

\subsubsection{RWS-QAOA QPU correlations}
In this section we describe the implementation of RWS-QAOA on the 54-qubit IQM Emerald QPU with a square-grid connectivity (see Fig.~\ref{fig:idea_sketch}{(A)}). We leveraged the fact that, upon solving the relaxed objective to determine the warm-start angles $\theta_i$, many of them take the extremal values $\theta_i\in\{0,\pi\}$. For $\theta_i=0$ ($\theta_i=\pi$) qubit $i$ is initialized in the computational basis state $\ket{0}$ ($\ket{1}$), and the associated mixer term reduces to $\pm Z_i$. Then, both the mixer and the phase separator act as diagonal unitaries, leaving these qubits frozen in their initial configurations throughout the circuit. We can therefore eliminate them from the cost function by substituting their frozen values, contributing a local field to each neighboring node. This leaves us with a reduced subgraph (see Fig.~\ref{fig:var-freezing}{(A)}), typically much smaller than the original graph $G$, and often split into several connected components. Only the reduced subgraph needs to be implemented on the QPU, with each connected component treated as a separate circuit. Therefore, the original problem size may exceed the number of available qubits (see Fig.~\ref{fig:emerald_maxcut}).

\begin{figure}
    \centering
    \includegraphics{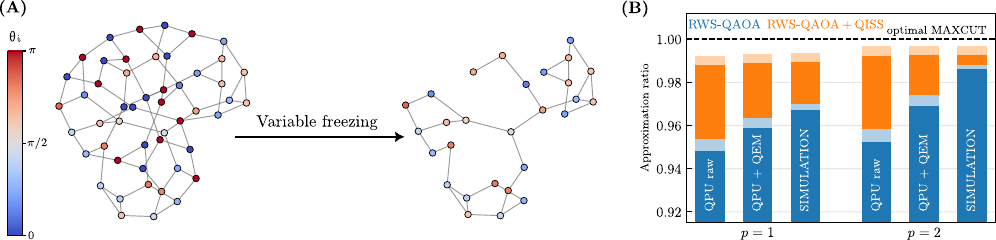}
    \caption{\textbf{Variable freezing and device results for RWS-QAOA.} (A) The initial angles $\theta_i$ obtained by minimizing Eq.~\eqref{eq:regularized-objective-rwsqaoa}. The variables with extremal angles $\theta_i=0$ or $\theta_i=\pi$ are frozen in the corresponding classical configurations, yielding a simplified problem. (B) Device results on the IQM Emerald QPU, averaged over 10 MaxCut instances on 3-regular graphs with $N=50$ nodes, for RWS-QAOA at depths $p=1,2$. Blue bars show the approximation ratio obtained directly from the RWS-QAOA correlators; the orange segment stacked on top shows the additional improvement from QISS. For each depth we consider three inputs (left to right): the raw QPU correlators, the error-mitigated (QEM) correlators, and noiseless simulation. The dashed line marks the optimal MaxCut. Shaded regions denote the standard error of the mean.}
    \label{fig:var-freezing}
\end{figure}

Figures~\ref{fig:var-freezing}{(B)} and~\ref{fig:emerald_maxcut} show the results obtained on the IQM Emerald QPU. At each system size we generated 10 random problem instances, computed the warm-start angles, and simplified the instances by freezing all variables with $\abs{\cos\theta_i}\geq 0.999$. Each connected component is run as a separate circuit, from which we collect $2000$ shots. To mitigate QPU noise we combine Pauli twirling (PT)~\cite{wallman2016,hashim2020} (32 twirls) with Zero Noise Extrapolation (ZNE)~\cite{Temme2017,Li2017} by linearly extrapolating from the twirled expectation values at noise levels $\lambda\in\{1,3\}$, implemented via gate folding.

\begin{figure}
    \centering
    \includegraphics[width=0.99\textwidth]{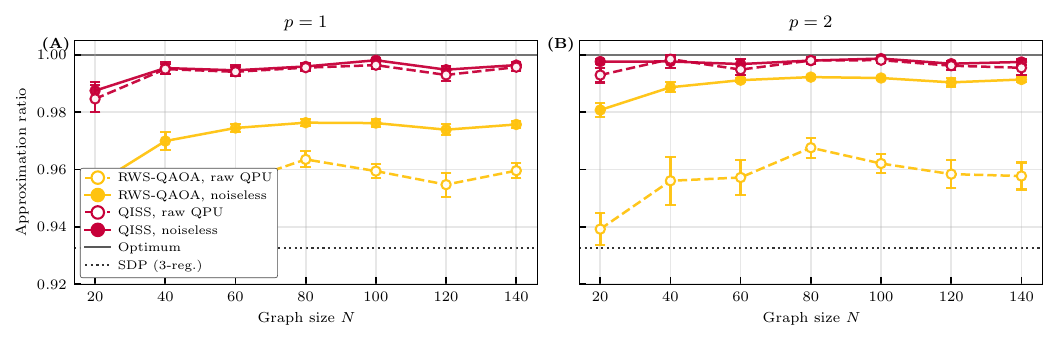}
    \caption{\textbf{QISS for MaxCut on 3-regular graphs on the  54-qubit IQM Emerald QPU.} Approximation ratios over 10 graph instances (and over QISS samples), for RWS-QAOA at $p=1$ (A) and $p=2$ (B). We compare the raw QPU correlators (open markers, dashed) with noiseless simulation (filled, solid), each with (red) and without (yellow) QISS post-processing. QISS lifts the approximation ratio to just below the optimum (solid line) and well above the SDP guarantee for $3$-regular graphs (dotted line), and its output is essentially insensitive to whether the correlators are noisy or noiseless. Error bars denote the standard error of the mean.}
    \label{fig:emerald_maxcut}
\end{figure}

In Fig.~\ref{fig:var-freezing}{(B)} we show the approximation ratios averaged over ten $3$-regular graphs with $N=50$ nodes, comparing the raw QPU correlators, the error-mitigated correlators, and the noiseless simulation, each with and without QISS post-processing. At the level of the bare RWS-QAOA estimator the raw QPU approximation ratio improves slightly with depth, and this improvement becomes more pronounced once quantum error mitigation (QEM) is applied. Notably, the $p=2$ error-mitigated correlators outperform the noiseless $p=1$ correlators, recovering the expected improvement with QAOA depth that the raw device data alone does not exhibit.

QISS consistently improves upon the approximation ratio of the correlators it is built from, lifting the solution quality to just below the optimum. Strikingly, the QISS output is essentially insensitive to device noise: the approximation ratios obtained from the raw QPU, the error-mitigated, and the noiseless correlators are nearly indistinguishable. The raw noisy correlators already carry enough structure to reach the same near-optimal cuts, making both the device noise and the QEM used to counter it largely irrelevant after post-processing.

The same behavior persists as the system size grows, as shown in Fig.~\ref{fig:emerald_maxcut}: across $N$, QISS consistently improves upon the RWS-QAOA correlators, drives the approximation ratio to just below the optimum and well above the SDP guarantee (dotted line)~\cite{Halperin2004}, and remains essentially insensitive to whether the correlators come from the raw QPU or the noiseless simulation. This holds up to $N=140$, the largest size at which the largest connected component of all ten instances fits on the Emerald QPU. The approximation ratio does not degrade with $N$ because, although the largest connected component grows on average, its simple structure (see Fig.~\ref{fig:var-freezing}{(A)}) keeps the transpiled circuit depth roughly constant.

\subsubsection{Comparison to classical algorithms}

We compare the results of QISS based on mean values from QAOA and from RWS-QAOA to two classical algorithms in Fig.~\ref{fig:comparison_maxcut}: (i) Simulated Annealing (SA) and (ii) a low-rank Burer--Monteiro (BM) implementation of the Goemans--Williamson SDP relaxation for MaxCut. As a simulated-annealing baseline~\cite{Kirkpatrick1983}, we use the classical \texttt{SimulatedAnnealingSampler} from D-Wave's Ocean \texttt{dwave.samplers} package~\cite{DWaveSamplers}. For each instance, we use 500 independent reads initialized from random spin configurations. The inverse-temperature schedule is geometric, with its range set automatically by the sampler from the coupling magnitudes. We use 1000 sweeps per read, with one full sequential Metropolis sweep over all spins per beta value. Note that this implies that in total we do two orders of magnitude more sweeps for SA compared to QISS. However, in both cases the number of sweeps is constant as a function of $N$, while the cost of one read scales as $O(N)$ for 3-regular graphs. We report averages over the 40 instances while keeping the best sample for each graph instance.

We also compare to the Burer--Monteiro rank-two relaxation heuristic~\cite{BurerMonteiroZhang2002}. We use the C\texttt{++} implementation provided by the MQLib library~\cite{Dunning2018} under the name \texttt{BURER2002} which is considered to be a state-of-the art classical solver~\cite{Dunning2018}. For each graph instance the solver is given a wall-clock budget of $2\,\mathrm{s}$ on a single core. Within this budget it performs repeated random restarts of the rank-two relaxation, each refined by a gradient-based optimization, randomized projection (hyperplane) rounding, and $1$-opt local search, and it returns the best cut encountered for every instance. 

Under a fixed compute budget, the purely classical samplers become budget-limited as the instances grow and their performance degrades for larger instances (see Fig.~\ref{fig:comparison_maxcut}). 

Combining RWS-QAOA with QISS post-processing yields a strong results in this comparison. However, the warm start inherits the classical relaxation it is built from, and therefore the performance can degrade with increasing system size as the budget is kept fixed, like the classical solvers. QAOA, by contrast, is a local algorithm whose correlations on bounded-degree graphs are set by size-independent local neighborhoods and are thus expected to stay roughly stable. The correlation signal passed from QAOA to QISS should therefore degrade more slowly than the quality of a Markov chain whose fixed sweep budget must cover a configuration space that grows with size, however the RWS relaxation could set a limit on how far this robustness extends.

 \begin{figure}
    \centering
    \includegraphics[width=0.5\textwidth]{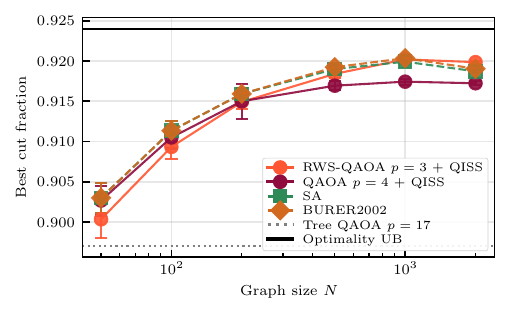}
    \caption{\textbf{QISS for MaxCut on 3-regular graphs compared to classical solvers.} The cut fractions averaged over the same 40 instances per $N$. All algorithms generate multiple candidate solutions during their run time. We keep the best and then average over the different instances, the error bars show the standard error of the mean.}
    \label{fig:comparison_maxcut}
\end{figure}

\subsection{Maximum Independent Set}

\begin{figure}
    \centering
    \includegraphics[width=\textwidth]{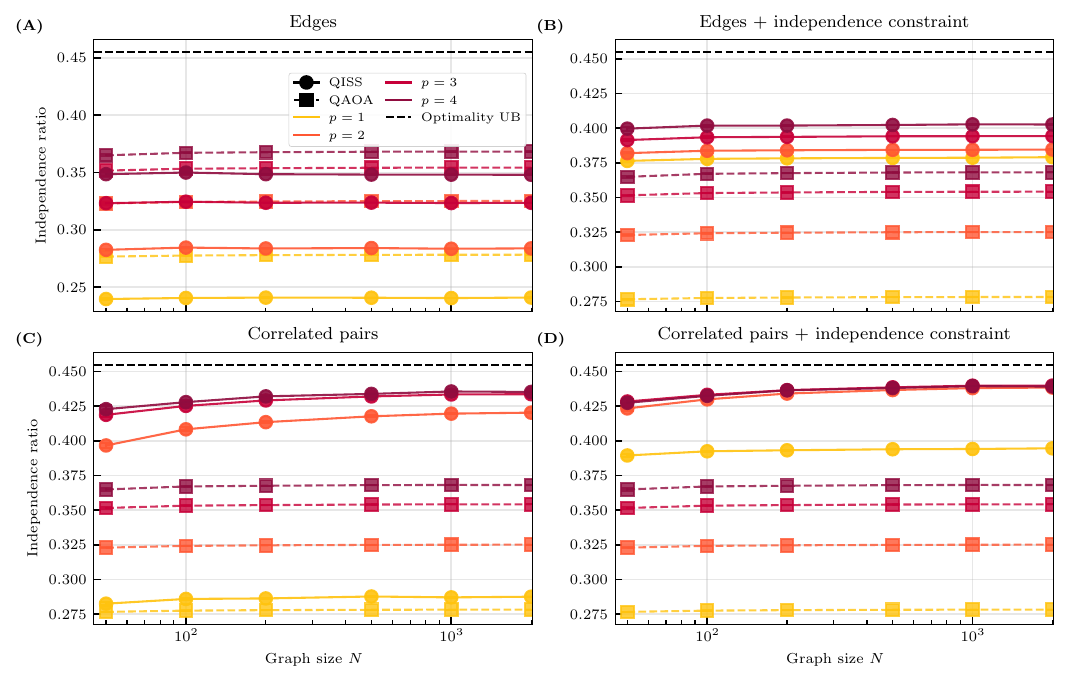}
    \caption{\textbf{QISS for MIS on 3-regular graphs based on QAOA.} The independence ratios obtained from QAOA and from the QISS for various QAOA depths $p$. (A), (B) Only the QAOA edge correlators and the polarizations, corresponding to factor collection $S$, are included in QISS. (C), (D) All correlated pairs and the polarizations, corresponding to factor collection $S'$, are included in QISS. The MIS problem has an independence constraint which we can enforce by simple post-processing of the obtained samples as described in the main text, panels (A), (C) and (B), (D) show respectively without and with postprocessing.}
    \label{fig:mis_means}
\end{figure}

For the MIS problem, the setup differs from MaxCut in two ways. First, the MIS Hamiltonian has no $\mathbb{Z}_2$ symmetry, so the one-body expectation values $\ev{Z_i}{\bm{\gamma}^{\star},\bm{\beta}^{\star}}$ are in general nonzero. Second, MIS is a constrained problem: a valid solution must be an independent set of $G$.

\subsubsection{QAOA correlations}

The absence of $\mathbb{Z}_2$ symmetry means that for pairs with non-overlapping light cones, $d_G(i,j)>2p$, the two-body correlators factorize as $\ev{Z_iZ_j}=\ev{Z_i}\ev{Z_j}$. We therefore supply QISS with the one-body expectations $\ev{Z_i}$ for all $i\in V$ together with the two-body correlators $\ev{Z_iZ_j}$ that do not factorize, i.e. those with $d_G(i,j)\le 2p$.

We use the same instance set as for MaxCut (40 random 3-regular graphs at each system size), compute the required depth-$p$ expectation values by tensor-network contraction, and take the angles $(\bm{\gamma}^{\star},\bm{\beta}^{\star})$ from Ref.~\cite{Wybo2025}. We then apply QISS (Sec.~\ref{sec:factor_sampling}) with either the edge set $S=E\cup V$ (Fig.~\ref{fig:mis_means}{(A, B)}) or the augmented set $S'=E'\cup V$ (Fig.~\ref{fig:mis_means}{(C, D)}), as defined in Sec.~\ref{sec:factor_distributions}. For each sample we compute the independence ratio $-\ev{H_{\mathrm{MIS}}^{\lambda=1}}{\bm{z}}/N$ from Eq.~\eqref{eq:H_MIS}.

In addition, MIS is a constrained problem, since a valid solution must correspond to an independent set of $G$. A candidate solution produced by QISS may still contain conflicting edges, i.e. edges for which both nodes are selected. Such samples can be converted into valid independent sets by an additional post-processing routine. We first identify conflicts by assigning to each vertex $i$ the conflict score
\begin{equation}
c(i)=\bigl|\{j\in N(i): z_i=z_j=+1\}\bigr|.
\end{equation}
Here $N(i)=\{j\in V \mid (i,j)\in E\}$ denotes the set of neighbors of vertex $i$. For every conflicting edge $(i,j)\in E$ with $z_i=z_j=+1$, we flip the endpoint with the largest conflict score, using the higher-index vertex as a tie breaker. After all conflicts have been removed, we greedily add vertices: for every vertex $i$ with $z_i=-1$, if
\begin{equation}
\forall j\in N(i):\ z_j=-1,
\end{equation}
we set $z_i=+1$. Applying this pruning routine improves the solution quality with respect to the MIS energy in Eq.~\eqref{eq:H_MIS}, with the unprocessed MIS energy serving as a lower bound. We report results with (Fig.~\ref{fig:mis_means}{(B,D)}) and without this constraint fixing procedure (Fig.~\ref{fig:mis_means}{(A,C)}). 
We observe that for MIS, unlike MaxCut, when only considering the edge set, the resampled results are worse than QAOA, see Fig.~\ref{fig:mis_means}{(A)}. So, in this case, even when the graphical model is locally tree like, the resampled averages drift away from the QAOA averages. This is a direct consequence of the absence of $\mathbb{Z}_2$ symmetry: with nonzero fields, one- and two-body factors overlap, so the surrogate no longer reproduces the QAOA correlators that determine the MIS energy. Only after pruning, the solution quality outperforms standalone QAOA, see Fig.~\ref{fig:mis_means}{(C)}.

\subsubsection{Algorithm comparison}

\begin{figure}
    \centering
    \includegraphics[width=0.5\textwidth]{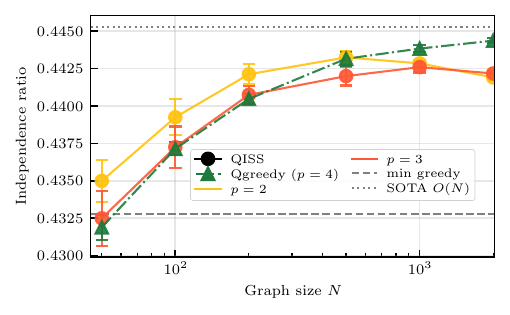}
    \caption{\textbf{QISS for MIS on 3-regular graphs compared to other algorithms} including the quantum-enhanced greedy of Ref.~\cite{Wybo2026}. For the QISS sampling solution we have kept the best out of 500 generated samples for each instance. The observed decline in the yellow and orange curves is a finite-size effect. The error bars show the standard error of the mean from the average over instances. The state-of-the-art value is taken from Ref.~\cite{Marino2020}.}
    \label{fig:mis_comparison}
\end{figure}

In Fig.~\ref{fig:mis_comparison}, we compare the independence ratios averaged over 40 instances at each system size $N$. For each instance, we took the best of the generated candidate solutions. We compare to classical baselines such as minimal greedy~\cite{Wormald1999} and the linear-prioritized search algorithm of Ref.~\cite{Marino2020}. 
In addition, we compare to the quantum-enhanced greedy (QGreedy) algorithm of Ref.~\cite{Wybo2026}. This algorithm builds an independent set greedily: at each step it adds the node with the largest expectation $\ev{Z_i}{\bm{\gamma}^{\star},\bm{\beta}^{\star}}$ to the independent set and then removes that node and its neighbors from the graph. Hence, it is an iterative approach in contrast to the sequential QISS approach. The performance of the QGreedy algorithm does not decline with system size $N$ and is scalable, in contrast to the slight decline of the resampled data with increasing $N$ (Fig.~\ref{fig:mis_comparison}). This is a finite-size effect arising from our reporting the best over a constant number of samples. The solution quality concentrates about its mean as the system grows, with fluctuations shrinking as $1/N$. The upper tail of the sample distribution, from which the best sample is drawn, therefore contracts toward the sample mean. We note that these sample means (averaged over instances) are shown in Fig.~\ref{fig:mis_means}{(D)}.

\section{Conclusion} \label{sec:concl}

We have introduced and benchmarked a classical post-processing method, Quantum-Informed Surrogate Sampling (QISS), that converts the local mean values produced by QAOA into a structured distribution, from which improved candidate solutions can be drawn by MCMC sampling. The construction requires no additional variational optimization and no further device access: the surrogate couplings follow directly from the measured quantum correlations, and sampling is performed classically.

Across our benchmarks, QISS improves substantially on the QAOA output at fixed circuit depth. This improvement hinges on the choice of correlator set. Feeding QISS only the edge correlators can reproduce the QAOA moments on a locally tree-like graph and yields no gain, while feeding it the augmented set of all correlators within the light cone, $d_G(i,j)\le 2p$, makes the factor graph non-tree-like and lets the surrogate move beyond the QAOA output. On MaxCut on 3-regular graphs, QISS built from depth $p>2$ correlators then exceeds vanilla tree QAOA at $p=17$, the largest depth with known tree angles, and is competitive with strong classical solvers such as Burer--Monteiro. 

Because the surrogate depends smoothly on the correlators, it is robust to noise in the estimated $\mu_{S_\alpha}$. We demonstrated this on the 54-qubit IQM Emerald QPU by showing that QISS recovers near-optimal solutions from the raw device correlators, on par with noiseless simulation, so that the device noise (and the error mitigation used to counter it) becomes largely irrelevant after post-processing.

The sampling step of QISS depends only on the correlators, not on the cost function, so the method extends readily to other combinatorial problems, including constrained ones. We demonstrate this explicitly on MIS, where QISS can be further improved by a light postprocessing step that enforces independence.

Finally, several extensions of QISS are natural. Our benchmarks use only weight-one and weight-two correlators on 3-regular graphs; incorporating higher-weight correlators, or structured encodings such as the Pauli-correlation encoding of Ref.~\cite{Sciorilli2025} could broaden its scope. Moreover, the applicability of QISS is also not limited to the NISQ era. Recent estimates for RWS-QAOA on MaxCut place the quantum--classical crossover at depth $p=6$ on few-thousand-node instances, requiring on the order of a million physical qubits at $90\%$ circuit fidelity~\cite{he2026regularizedwarmstartedquantumapproximate}. A post-processing layer that recovers comparable quality from lower-depth correlators could bring this crossover closer.

\section{Acknowledgments}
We thank Alessio Calzona, Martin Leib and Fedor \v{S}imkovic for helpful discussions and valuable feedback. We also acknowledge our colleagues at IQM for their support and for providing a collaborative research environment.

\appendix

\section{The Markov process satisfies detailed balance and ergodicity} \label{app:MCMC_proofs}

\begin{proposition}
Consider the probability distribution on $\{-1,+1\}^N$
\begin{equation} 
    P(\bm{z})
    \propto
    \prod_{S_{\alpha}\in S} 
    \bigl(1+\mu_{S_{\alpha}} \,\chi_{S_{\alpha}}(\bm{z})\bigr),
\end{equation}
where $\chi_{S_{\alpha}}(\bm{z})$ is any function of the spins $\{z_i : i \in S_\alpha\}$ and the product is strictly positive for all $\bm{z}$. 
Let the single-site sampler be defined as follows: at each step choose a site $i$ with probability $w_i>0$, $\sum_i w_i=1$, and resample $z_i$ from the conditional $P(z_i \mid \bm{z}_{-i})$. Then the resulting Markov chain satisfies detailed balance with respect to $P$.
\end{proposition}

\begin{proof}
Fix a site $i$ and denote by $P_i(\bm{x}\to\bm{y})$ the transition kernel of a single-site update at $i$:
\[
P_i(\bm{x}\to\bm{y})
=
\begin{cases}
P\bigl(y_i \mid \bm{x}_{-i}\bigr), & \text{if } \bm{x}_{-i}=\bm{y}_{-i},\\[2pt]
0, & \text{otherwise.}
\end{cases}
\]
We first show detailed balance for $P_i$:
\begin{equation}\label{eq:db-single}
    P(\bm{x})\,P_i(\bm{x}\to\bm{y})
    = 
    P(\bm{y})\,P_i(\bm{y}\to\bm{x})
    \qquad \forall\,\bm{x},\bm{y}.
\end{equation}

If $\bm{x}_{-i}\neq\bm{y}_{-i}$ then $P_i(\bm{x}\to\bm{y})=P_i(\bm{y}\to\bm{x})=0$, so \eqref{eq:db-single} holds trivially.  
Otherwise, write $\bm{x}=(x_i,\bm{u})$ and $\bm{y}=(y_i,\bm{u})$ with common environment $\bm{u}=\bm{x}_{-i}=\bm{y}_{-i}$. By the definition of conditional probability,
\[
P_i(\bm{x}\to\bm{y})
= P(y_i \mid \bm{u})
= \frac{P(y_i,\bm{u})}{\sum_{z_i'=\pm1} P(z_i',\bm{u})},
\]
and similarly
\[
P_i(\bm{y}\to\bm{x})
= P(x_i \mid \bm{u})
= \frac{P(x_i,\bm{u})}{\sum_{z_i'=\pm1} P(z_i',\bm{u})}.
\]
Therefore
\[
P(\bm{x})\,P_i(\bm{x}\to\bm{y})
= P(x_i,\bm{u}) \,
  \frac{P(y_i,\bm{u})}{\sum_{z_i'} P(z_i',\bm{u})},
\qquad
P(\bm{y})\,P_i(\bm{y}\to\bm{x})
= P(y_i,\bm{u}) \,
  \frac{P(x_i,\bm{u})}{\sum_{z_i'} P(z_i',\bm{u})},
\]
and the right-hand sides are equal by symmetry of the numerator, proving \eqref{eq:db-single}.  
Note that we only used the fact that the conditional probablilities of $P$ can be efficiently computed: the explicit factor form \eqref{eq:factors} is irrelevant for detailed balance itself.

The transition kernel of the MCMC process is the convex combination
\[
P(\bm{x}\to\bm{y}) = \sum_{i=1}^N w_i\, P_i(\bm{x}\to\bm{y}).
\]
Using \eqref{eq:db-single},
\[
P(\bm{x}) P(\bm{x}\to\bm{y})
= \sum_i w_i\, P(\bm{x}) P_i(\bm{x}\to\bm{y})
= \sum_i w_i\, P(\bm{y}) P_i(\bm{y}\to\bm{x})
= P(\bm{y}) P(\bm{y}\to\bm{x}),
\]
so $P(\bm{x}\to\bm{y})$ satisfies detailed balance with respect to $P$.
\end{proof}

Because $P(\bm{z})>0$ for all $\bm{z}$, every single-site update has strictly positive probability to flip any spin in any configuration, and also to keep it unchanged. This yields irreducibility (via sequences of single-spin flips) and aperiodicity (nonzero self-loop at every state), hence ergodicity.

\begin{proposition}[Ergodicity]
Consider the distribution
\begin{equation}
    P(\bm{z})
    \propto
    \prod_{S_{\alpha}\in S} 
    \bigl(1+\mu_{S_{\alpha}} \chi_{S_{\alpha}}(\bm{z})\bigr),
    \qquad \bm{z}\in\{-1,+1\}^N,
\end{equation}
and assume $p(\bm{z})>0$ for all configurations $\bm{z}$.  
Let the single-site sampler be defined as follows: at each step, choose a site $i$ with probability $w_i>0$, $\sum_i w_i=1$, and resample $z_i$ from $P(z_i\mid\bm{z}_{-i})$. Then the resulting Markov chain on $\{-1,+1\}^N$ is ergodic, i.e. irreducible and aperiodic, and hence converges to $P$ from any initial configuration.
\end{proposition}

\begin{proof}
\noindent\textbf{Irreducibility.}
Take any two configurations $\bm{x},\bm{y}\in\{-1,+1\}^N$. There is a path from $\bm{x}$ to $\bm{y}$ that flips the spins one by one:
\[
\bm{x} = \bm{z}^{(0)} \to \bm{z}^{(1)} \to \cdots \to \bm{z}^{(K)} = \bm{y},
\]
where each $\bm{z}^{(k+1)}$ differs from $\bm{z}^{(k)}$ at exactly one site $i_k$.  
For the random-scan kernel, the probability of the transition $\bm{z}^{(k)}\to \bm{z}^{(k+1)}$ is
\[
P\bigl(\bm{z}^{(k)} \to \bm{z}^{(k+1)}\bigr)
= w_{i_k} \, P\bigl(z^{(k+1)}_{i_k} \mid \bm{z}^{(k)}_{-i_k}\bigr).
\]
By assumption $P(\bm{z})>0$ for all $\bm{z}$, so every conditional $P(z_i\mid\bm{z}_{-i})$ assigns strictly positive probability to both $z_i=\pm1$. Hence
\[
w_{i_k} > 0
\quad \text{and} \quad
P\bigl(z^{(k+1)}_{i_k} \mid \bm{z}^{(k)}_{-i_k}\bigr) > 0
\quad\Rightarrow\quad
P\bigl(\bm{z}^{(k)} \to \bm{z}^{(k+1)}\bigr) > 0.
\]
The product of these positive probabilities along the path is also positive, so the chain can reach $\bm{y}$ from $\bm{x}$ with nonzero probability in finitely many steps. Since $\bm{x},\bm{y}$ were arbitrary, the chain is irreducible.\\

\noindent\textbf{Aperiodicity.}
For any configuration $\bm{z}$, consider the probability of remaining in $\bm{z}$ in one step:
\[
P(\bm{z}\to\bm{z})
= \sum_{i=1}^N w_i \, P(z_i \mid \bm{z}_{-i}).
\]
Again by strict positivity of $P$, we have $P(z_i\mid \bm{z}_{-i})>0$ for all $i$, and $w_i>0$ by assumption. Thus each term $w_i P(z_i\mid \bm{z}_{-i})$ is positive, so $P(\bm{z}\to\bm{z})>0$.  
A Markov chain on a finite state space with a strictly positive self-loop at every state has period $1$ at every state, hence is aperiodic. \\

The chain is irreducible and aperiodic and, by detailed balance (proved separately), has $P$ as stationary distribution. Therefore it is ergodic and converges to $P$ from any initial $\bm{z}$.
\end{proof}

\section{Thermal correlations}\label{app:thermal}
\begin{figure}
    \centering
    \includegraphics[width=\textwidth]{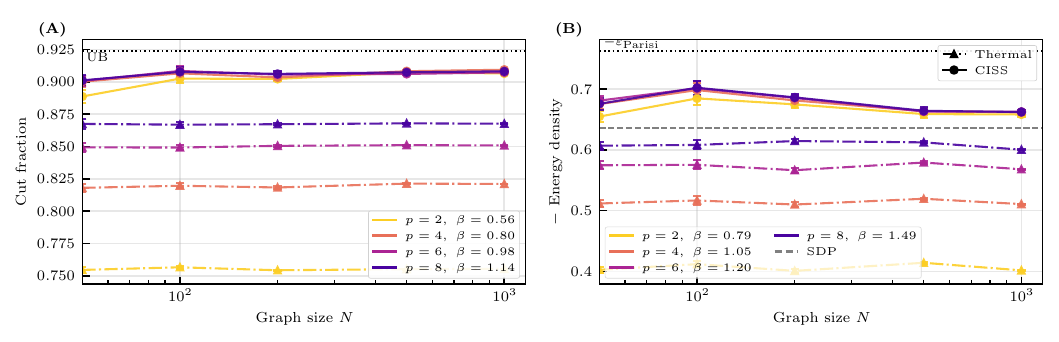}
    \caption{\textbf{Resampling from thermal correlations for the MaxCut and SK model.}
    Candidate solutions obtained by sampling the surrogate factor model (CISS, Sec.~\ref{sec:factor_distributions}) built from thermal two-point correlations, estimated by parallel-tempering MCMC at inverse temperature $\beta$. In both panels $\beta$ is fixed by matching the thermal energy density to that of depth-$p$ QAOA. Dash-dotted lines (triangles) show the raw thermal input and solid lines show the resampled output. Error bars show the standard error of the mean. \textbf{(A)} MaxCut on $3$-regular graphs. The dotted line marks the optimality upper bound. The resampled values (solid) lie close to UB and depend only weakly on $\beta$. \textbf{(B)} SK model. The dashed line marks the Goemans--Williamson SDP value and the dotted line the Parisi energy. The resampled values exceed the SDP benchmark but saturate well below Parisi, and colder input does not improve them further. The improvement upon the raw input diminishes with increasing system size.}
\label{fig:sk}
\end{figure}

So far we have used correlations obtained from QAOA states as the input to QISS. To probe to what extent the resampling procedure depends on the specific structure of the input correlations, rather than on generic properties of the used two-point functions, we here perform experiments with correlations obtained from a classical thermal state at different inverse temperatures $\beta$. We refer to this variant as Correlation-Informed Surrogate Sampling (CISS).

Concretely, for a problem instance with cost Hamiltonian $H_C$ we take the Gibbs state $P_\beta(\bm z)\propto e^{-\beta H_C}$ and use its two-point correlations as input to the factor model of Eq.~\eqref{eq:factors}, instead of the QAOA expectation values. The correlations are estimated by parallel-tempering Markov-chain Monte Carlo~\cite{Hukushima1996}, using $56$ replicas linearly spaced over $\beta\in[\beta_{\min},\beta_{\max}]$, with $\beta_{\min}=0.4$ and $\beta_{\max}=3.8$, $3500-5000$ burn-in sweeps followed by $6500-14000$ measurement sweeps and replica-exchange attempts every sweep. Correlations are accumulated at the target $\beta$ replica. We monitor equilibration through the integrated autocorrelation time of the energy.

We consider the thermal correlations for two problems: MaxCut on random $3$-regular graphs and the Sherrington--Kirkpatrick (SK) model. The SK model is the fully-connected Ising spin glass with Hamiltonian
\begin{equation}
    H_{SK}(\bm{z}) = -\sum_{i<j} J_{ij}\, z_i z_j
\end{equation}
where the couplings $J_{ij}$ are i.i.d. Gaussian with zero mean and variance $1/N$. The two problems probe complementary structural regimes: the $3$-regular instances are sparse and of bounded degree, whereas the SK model is dense and fully connected. 

For MaxCut on $3$-regular graphs Fig.~\ref{fig:sk}(A), resampling from thermal correlations still improves substantially upon the raw thermal cut fractions. In contrast to the QAOA case, however, the improvement is not monotonic in the energy of the input correlations: the resampled cut fraction saturates at the lowest depths so that colder thermal input no longer yields better solutions. This is not a ceiling of the surrogate sampling itself, since resampling from the RWS-QAOA correlations reaches higher cut fractions (see Fig.~\ref{fig:maxcut_ws_means}). Rather, the surrogate samples a pairwise product model at its own fixed effective temperature, so sharpening the input couplings does not cool the sampler, and the plateau reflects a property of the (glassy) thermal input rather than of the resampling. A plausible mechanism is that in the glassy phase, the thermal state fragments into many competing states, whose averaged correlations no longer reflect a single coherent assignment.

Similarly, for the SK model shown in Fig.~\ref{fig:sk}(B), resampling again improves markedly upon the raw thermal input, lifting the energy densities above the SDP value, but saturating well below the Parisi energy. Here the resampled values for $\beta\gtrsim1$ again collapse onto a single curve, so that colder input is no longer converted into better solutions; this saturation sets in near the spin-glass transition $\beta_c=1$~\cite{Parisi1979}.

\bibliography{biblio}

\end{document}